\documentclass{llncs}
\usepackage{latexsym}
\usepackage{amsmath,amssymb,amscd}
\usepackage{bussproofs}
\usepackage{color}
\usepackage{proof}
\usepackage{multicol}
\usepackage[english]{babel}
\usepackage[shortlabels]{enumitem}
\usepackage{tikz}
\usetikzlibrary{matrix}
\usetikzlibrary{positioning}

\newcommand{\skipall}[1]{}
\usepackage{epsfig}
\usepackage{multicol}
\usepackage{tikz}
\usetikzlibrary{arrows,automata}
\usepackage{array}
\usepackage{float}

\newcommand{\vN}{\vdash_{\mathcal{NE}_p}}
\newcommand{\vL}{\vdash_{\mathcal{LE}_p}}

\input{mycommands.sty}

\date{}

\title{On an ecumenical natural deduction with {\em stoup} - Part I: The propositional case}
\author{Luiz Carlos Pereira\inst{1}\thanks{Pereira is supported by CAPES (PRINT and COFECUB) and CNPq.} \and Elaine Pimentel\inst{2}}
\institute{PUC-Rio/UERJ/CNPq \and UCL/UFRN\\ luiz@inf.puc-rio.br \; e.pimentel@ucl.ac.uk}

\begin{document}
\maketitle
\setcounter{secnumdepth}{3}
\section{Introduction}
Natural deduction systems, as proposed by Gentzen~\cite{gentzen1969} and further studied by Prawitz~\cite{prawitz1965}, is one of the most well known proof-theoretical frameworks. Part of its success is based on the fact that natural deduction rules present a simple characterization of logical constants, especially in the case of intuitionistic logic. However, there has been a lot of criticism on extensions of the intuitionistic set of rules in order to deal with classical logic. Indeed, most of such extensions add, to the usual introduction and elimination rules, extra rules governing  negation. As a consequence, several meta-logical properties, the most prominent one being {\em harmony}, are lost.\footnote{A logical connective is called {\em harmonious} in a certain proof system if there exists a certain balance between the rules defining it. For example, in natural deduction based systems, harmony is ensured when introduction/elimination rules do not contain insufficient/excessive amounts of information~\cite{DBLP:conf/ictac/Diaz-CaroD21}.}

In~\cite{DBLP:journals/Prawitz15}, Dag Prawitz proposed a
natural deduction {\em ecumenical system}, where classical logic and intuitionistic logic are codified in the same system. In this system, 
the classical logician and the intuitionistic logician would share the universal quantifier, conjunction, negation and the constant for the absurd ($\forall,\wedge,\neg,\bot$), but they would each have their own existential quantifier, disjunction and implication, with different meanings ($\exists_j,\vee_j,\to_j$, where $j\in\{i,c\}$ for the intuitionistic and classical versions, respectively). Prawitz' main idea is that these different meanings are given by a semantical framework that can be accepted by both parties.  

In his ecumenical system, Prawitz recovers the harmony of rules, but the rules for the classical operators do not satisfy {\em separability}~\cite{Murzi2018}. In fact, the classical rules are not {\em pure}, in the sense that 
negation 
is used in the definition of the introduction and elimination rules for the classical operators. 

For example, the rules for $\vee_{c}$ are defined as

\begin{prooftree}
\AxiomC{$[\neg A,\neg B]$}
\noLine
\UnaryInfC{$\Pi$}
\noLine
\UnaryInfC{$\bot$}
\RightLabel{$\vee_{c}$-int}
\UnaryInfC{$A \vee_{c} B$}
\DisplayProof
\qquad
\AxiomC{$A \vee_{c} B$}
\AxiomC{$\neg A$}
\AxiomC{$\neg B$}
\RightLabel{$\vee_{c}$-elim}
\TrinaryInfC{$\bot$}
\end{prooftree}
The situation is not different in the case of the definition of left and right rules for these classical operators in a {\em sequent calculus} codification, as presented in~\cite{DBLP:journals/synthese/PimentelPP21}. The rules for $\vee_{c}$, for example, are defined as
\begin{prooftree}
\AxiomC{$\Gamma , \neg A , \neg B \Rightarrow \bot$}
\RightLabel{$\vee_{c}$-R}
\UnaryInfC{$\Gamma \Rightarrow A \vee_{c} B$}
\DisplayProof
\qquad
\AxiomC{$\Gamma \Rightarrow \neg A$}
\AxiomC{$\Delta \Rightarrow \neg B$}
\RightLabel{$\vee_{c}$-L}
\BinaryInfC{$A \vee_{c} B, \Gamma ,\Delta \Rightarrow \bot$}
\end{prooftree}

There are many ways of proposing pure, harmonic natural deduction systems for (propositional) classical logic. 
Indeed, Murzi~\cite{Murzi2018} proposes a new set of rules for classical logical operators based on absurdity as a punctuation mark, and higher-level rules~\cite{DBLP:journals/sLogica/Schroeder-Heister14}. 
D'Agostino~\cite{DBLP:conf/birthday/DAgostino05}, on the other hand, brings a totally sifferent approach, presenting a theory of classical natural deduction that makes a distinction between operational rules, governing the use of logical operators, and structural rules dealing with the metaphysical assumptions governing the (classical) notions of truth and falsity, namely the principle of bivalence and the principle of non-contradiction.

A complete different approach is presented in~\cite{DBLP:conf/birthday/GabbayG05}, where Michael and Murdoch Gabbay present the natural deduction version of Dov Gabbay's {\em Restart} rule
\[
\infer[Restart]{B}{A}
\]
with the side-condition that, below every occurrence of $Restart$ from $A$ to $B$, there is (at least) one occurrence of $A$. The intended meaning is that $B$ is a new start to a line of reasoning concluding $A$. For example, in the derivation of the Peirce's Law
\begin{prooftree}
\AxiomC{$[(A \to B) \to A)]$}
\AxiomC{$[A]$}
\RightLabel{$Restart^\star$}
\UnaryInfC{$B$}
\RightLabel{$\to$-int}
\UnaryInfC{$A\to B$}
\RightLabel{$\to$-elim}
\BinaryInfC{$A^\dagger$}
\RightLabel{$\to$-int}
\UnaryInfC{$((A \to B) \to A) \to_{c} A)$}
\end{prooftree}
the restart at ${}^\star$ is justified at $\dagger$.

Similar to the Gabbays, Restall's $Alt$ rule~\cite{restall} 
\[
\infer[Alt, \downarrow A]{B}{\deduce{A}{\deduce{}{\Pi}}}
\]
has the following interpretation: Having a proof of $A$, one can set $A$ aside and consider some alternative conclusion, $B$, while $A$ is added to the collection of alternatives current at this point of the proof. Hence, the proof of Peirce's Law would have the form
\begin{prooftree}
\AxiomC{$[(A \to B) \to A)]$}
\AxiomC{$[A]$}
\RightLabel{$Alt, \downarrow A$}
\UnaryInfC{$B$}
\RightLabel{$\to$-int}
\UnaryInfC{$A\to B$}
\RightLabel{$\to$-elim,$\uparrow A$}
\BinaryInfC{$A$}
\RightLabel{$\to$-int}
\UnaryInfC{$((A \to B) \to A) \to_{c} A)$}
\end{prooftree}

In this paper, we propose a different approach adapting, to the natural deduction framework,  Girard's mechanism of \textit{stoup}~\cite{DBLP:journals/mscs/Girard91}.  This will allow the definition of a pure harmonic natural deduction system $\mathcal{LE}_{p}$ for the propositional fragment of  Prawitz' ecumenical logic, where $Restart$ and $Alt$ appear as special cases of the use of {\em stoup}.\footnote{It should be noted that this idea  appears somewhat hidden in~\cite{restall} for the propositional classical case, where Restall uses sequent style
$X\rest A; Y$ (the `score'), where $X$ represent the undischarged assumptions, $A$ the current conclusion, and $Y$ the {\em alternatives}.}

\section{Ecumenical natural deduction system}\label{sec:ec}

The language $\Lscr$ used for ecumenical systems is described as follows. We will use a subscript $c$ for the classical meaning and $i$ for the intuitionistic, dropping such subscripts when formulae/connectives can have either meaning. 

Classical and intuitionistic  n-ary predicate symbols ($p_{c}, p_{i},\ldots$) co-exist in $\Lscr$ but have different meanings. 
The neutral logical connectives $\{\bot,\neg,\wedge,\forall\}$ are common for classical and intuitionistic fragments, while $\{\iimp,\ivee,\iexists\}$ and $\{\cimp,\cvee,\cexists\}$ are restricted to intuitionistic and classical interpretations, respectively. In order to avoid clashes of variables we make use of a denumerable set $a, b,\ldots$ of special variables called {\em parameters}, which do not appear quantified.

In Fig.~\ref{fig:NE} we present $\mathcal{NE}$ Prawitz' original natural deduction first-order ecumenical system. In the rules for quantifiers, the notation $A(a/x)$  stands for the substitution of $a$ for every (visible) instance of $x$ in $A$. In the rules $\exists_i\mbox{-elim},\forall\mbox{-int}$, $a$ is a fresh parameter, \ie, it does not occur free in any assumption that $A,B$ depends on (apart from the assumption eliminated by $\exists_{i}$-elim).

\begin{figure}[htp]
{\sc Intuitionistic rules}
\[
\begin{array}{lc@{\quad}l}
\infer[{\iimp\mbox{-elim}}]{B}{A\iimp B & A}
&
\infer[{\iimp\mbox{-int}}]{ A\iimp B}
{\deduce{B}{\deduce{}{\deduce{\Pi}{\deduce{}{[A]}}}}} 
&
\infer[{\vee_i\mbox{-elim}}]{C}{ A\vee_i B & \deduce{C}{\deduce{}{\deduce{\Pi_1}{\deduce{}{[A]}}}}
 &\deduce{C}{\deduce{}{\deduce{\Pi_2}{\deduce{}{[B]}}}}} 
\\[5pt]
\infer[{\vee_i\mbox{-int}_j}]{ A_1\vee_i A_2}{ A_j}
&
\infer[\exists_i\mbox{-elim}]{ B}{ \exists_ix.A &  \deduce{B}{\deduce{}{\deduce{\Pi}{\deduce{}{[A(a/x)]}}}}}
&
\infer[\exists_i\mbox{-int}]{\exists_ix.A }{ A(a/x)}
\end{array}
\]
{\sc Classic rules}
\[
\begin{array}{lc@{\quad}l}
\infer[{\cimp\mbox{-elim}}]{\bot}{ A\cimp B & A & \neg B}
&
\infer[{\cimp\mbox{-int}}]{A\cimp B}
{\deduce{\bot}{\deduce{}{\deduce{\Pi}{\deduce{}{[A,\neg B]}}}}} 
&
\infer[{\vee_c\mbox{-elim}}]{\bot}{A\vee_c B & \neg A &  \neg B} 
\\[5pt]
\infer[{\vee_c\mbox{-int}}]{A\vee_c B}{\deduce{\bot}{\deduce{}{\deduce{\Pi}{\deduce{}{[\neg A,\neg B]}}}}} 
&
\infer[\exists_c\mbox{-elim}]{\bot}{ \exists_cx.A &  \forall x.\neg A}
&
\infer[\exists_c\mbox{-int}]{ \exists_cx.A}{\deduce{\bot}{\deduce{}{\deduce{\Pi}{\deduce{}{[\forall x.\neg A]}}}}}
\\[5pt]
\infer[p_c\mbox{-elim}]{ \bot}
{
p_c &  \neg p_i
}
&
\infer[p_c\mbox{-int}]{ p_c}
{
\deduce{\bot}{\deduce{}{\deduce{\Pi}{\deduce{}{[\neg p_i]}}}}
}
\end{array}
\]
{\sc Neutral rules}
\[
\begin{array}{lc@{\qquad}lc@{\qquad}lc@{\qquad}l}
\infer[{\wedge\mbox{-elim}_j}]{ A_j}{A_1\wedge A_2} 
& &
\infer[{\wedge\mbox{-int}}]{A \wedge B}{A \quad  B} 
& &
\infer[{\neg\mbox{-elim}}]{\bot}{ A &  \neg A}
& &
\infer[{\neg\mbox{-int}}]{\neg A}{\deduce{\bot}{\deduce{}{\deduce{\Pi}{\deduce{}{[A]}}}}}
\\[5pt]
\infer[\bot\mbox{-elim}]{A}{ \bot}
& &
\infer[\forall\mbox{-elim}]{A(a/x) }{\forall x.A}
& &
\infer[\forall\mbox{-int}]{ \forall x.A}{ A(a/x)}
\end{array}
\]

\caption{Ecumenical natural deduction system $\mathcal{NE}$.  In rules 
$\forall\mbox{-int}$ and $\exists_i\mbox{-elim}$, the parameter $a$ is fresh.}\label{fig:NE}
\end{figure}

The propositional fragment of the natural deduction Ecumenical system proposed by Prawitz (here called $\mathcal{NE}_{p}$) has been proved normalizing, sound and complete with respect to  intuitionistic logic's Kripke semantics in~\cite{luiz17}. 

The rules for intuitionistic implication are the traditional ones, while the rules for classical implication make sure that $A\to_{c} B$ is treated as $\neg A\lor_c B$, its classical rendering. The surprising facts are that (i) one can have a single constant for absurdity $\bot$ (instead of two, one intuitionistic and one classical, taking that absurd as the unit of disjunction, of which we have two variants) and (ii) that the intuitionistic and classical negations coincide.
If negation was simply implication into false (as it is the case for intuitionistic negation) one might expect two negations, one intuitionistic and one classical.

\section{Sequent Calculus with \textit{stoup}}\label{sec:sc}

$\LC$~\cite{DBLP:journals/mscs/Girard91} is a sequent system for classical logic that separates the rules for {\em positive} and {\em negative} formulas, being a precursor of the notion of {\em focusing} in sequent systems~\cite{andreoli92jlc}. The polarity of a formula in $\LC$ is determined by its outermost connective and the polarity of its subformulas. For example, atoms and unities are always positive and, if $P,Q$ are positive formulas, then  $P\wedge Q$ is also positive. The table of polarities can be checked in~\cite{DBLP:journals/mscs/Girard91}, page 10.

In $\LC$,  sequents have the form
$\seq \Delta ; \Sigma$, where $\Delta , \Sigma$ are multisets of formulas, with $\Sigma$, called the {\em stoup}, containing {\em at most} one formula. The main idea is that the {\em stoup} controls the rule applications, in the sense that a positive active formula in the conclusion of a rule is always placed there, while active negative formulas are handled in the classical context.
Characteristic examples are the rules for the conjunction of positive/negative formulas  
\[
\infer[\wedge_p]{\seq \Gamma,\Delta ; P\wedge Q}{\seq \Gamma ; P & \seq \Delta ; Q}\qquad
\infer[\wedge_n]{\seq \Gamma,M\wedge N;\Sigma}{\seq \Gamma, N;\Sigma & \seq \Gamma,M ; \Sigma}
\]
Observe that they also have a multiplicative/additive flavor.

The idea of focusing is also present in $\LC$, with the {\em dereliction} and {\em store} rules
\[
\infer[der]{\seq \Gamma,P;\cdot}{\seq \Gamma;P}\qquad \qquad
\infer[store]{\seq \Gamma;N}{\seq \Gamma,N;\cdot}{}
\]
While in $der$ positive formulas can be chosen to be focused on, in $store$ negative formulas are stored in the classical context, in a bottom-up reading of rules. This enables for a {\em two-phase} proof construction, where the focused formula $P$ is systematically decomposed until reaching a leaf or a negative sub-formula $N$. In this last case, focusing is lost and $N$ is stored, allowing for the beginning of a new focused phase.

Finally, due to polarities, $\LC$ has two admissible cut rules
\[
\infer[p-cut]{\seq \Gamma,\Delta;\Sigma}{\seq \Gamma;P & \seq \neg P,\Delta;\Sigma}\qquad\qquad
\infer[n-cut]{\seq \Gamma,\Delta;\Sigma}{\seq \Gamma,N;\cdot & \seq \neg N,\Delta;\Sigma}
\]

In $\LC$, the sequent $\seq \Delta ; \Sigma$ has
an intuitionistic interpretation: $\neg \Delta\seq \bot$ if $\Sigma$ is empty and 
$\neg \Delta\seq A$ if $\Sigma=A$. That is, the context $\Delta$ makes the classical information persistent, via an implicit double negation elimination. This implies that sequents with empty {\em stoup} have also a classical interpretation, using \eg\ G\"{o}del's double negation translation. Sequents with non-empty stoup do not have a classical interpretation, as discussed in~\cite{DBLP:journals/mscs/Girard91}.

Hence one could say that sequents with stoup have a certain ecumenical flavor: formulas with intuitionistic behavior are identified as being {\em positive}, while formulas with classical 
behavior are identified as being {\em negative}.

In this work, we will carry out a similar idea under the spectrum of Prawitz' ecumenical natural deduction system.
While it has some similarities with Girard's original proposal, our system will not consider polarities, and all the conclusion formulas of introduction/elimination rules will be placed in the {\em stoup}.

\section{Natural Deduction with \textit{stoup} - the propositional system $\mathcal{LE}_{p}$}

We  will now incorporate the notion of {\em stoups} to natural deduction in the case of propositional logic, showing its natural connection to the ecumenical setting.

Let the expression  $\Delta;\Sigma$ denote an {\em stoup with a context} (abreviated as \textit{stp-c}), an extension of natural deduction formulas, where $\Sigma$ is the {\em stoup} and  $\Delta$ is its accompanying  context (called {\em alternatives} in~\cite{restall}). As for the case of $\LC$, the {\em stoup} will carry the intuitionistic (positive) and neutral information, while the context accumulates the classical information related to it.

In the following, we will construct ecumenical introduction and elimination rules over the {\em stoup} in a step-by-step manner, justifying all our choices.

\subsection{Intuitionistic operators}
\subsubsection*{Implication.}
The following result, proved in~\cite{DBLP:journals/synthese/PimentelPP21} for the sequent calculus ecumenical system, states that logical consequence in $\mathcal{NE}$ is interpreted intuitionistically. 
\begin{theorem}\label{prop:ev}
Let $\Gamma$ be a set of ecumenical formulas. Then 
$B$ is provable from $\Gamma$ in ecumenical logic iff $\bigwedge\Gamma\iimp B$ is provable in $\mathcal{NE}$.
\end{theorem}
Hence the rule
\[
\infer[{\iimp\mbox{-int}}]{ A\iimp B}
{\deduce{B}{\deduce{}{\deduce{\Pi}{\deduce{}{[A]}}}}}
\]
induces the rule with {\em stoup}
\[
\infer[{\iimp\mbox{-int}}]{\Delta;A\iimp B}{\deduce{\Delta;B}{\deduce{}{\deduce{\Pi}{\deduce{}{[\cdot;A] & \Gamma}}}}}
\]
where the {\em stoup} is preserved in discarded assumptions. Note that all the information $\Gamma$ about contexts should be remembered.

For the implication elimination rule, observe that different {\em stoups} carry different contexts, so we will have the multiplicative version of the rule, combining the classical information
\begin{prooftree}
\AxiomC{$\Delta_{1};A \to_{i} B$}
\AxiomC{$\Delta_{2};A$}
\RightLabel{$\to_{i}$-elim}
\BinaryInfC{$\Delta_{1},\Delta_{2};B$}
\end{prooftree}

\subsubsection*{Disjunction.}
The rule for introduction is
\begin{prooftree}
\AxiomC{$\Delta ;A_i$}
\RightLabel{$\vee_{i}$-int}
\UnaryInfC{$\Delta ;A_1 \vee_{i} A_2$}
\end{prooftree}

\noindent while the elimination rule combines the all the context information

\begin{prooftree}
\AxiomC{$\Gamma_1$}
\noLine
\UnaryInfC{$\Pi_{1}$}
\noLine
\UnaryInfC{$\Delta_{1};A \vee_{i} B$}
\AxiomC{$[\cdot ;A]$}
\AxiomC{$\Gamma_2$}
\noLine
\BinaryInfC{$\Pi_{2}$}
\noLine
\UnaryInfC{$\Delta_{2};C$}
\AxiomC{$[\cdot ;B]$}
\AxiomC{$\Gamma_3$}
\noLine
\BinaryInfC{$\Pi_{3}$}
\noLine
\UnaryInfC{$\Delta_{3};C$}
\RightLabel{$\vee_{i}$-elim}
\TrinaryInfC{$\Delta_{1}, \Delta_{2}, \Delta_{3};C$}
\end{prooftree}

\subsection{Classical operators}
\subsubsection*{Implication.} Observe that the negated assumptions in $\mathcal{NE}$ will correspond to the classical counterpart of the stoup, so the consequent of the implication will be stored in this context. The introduction rule 
\[
\infer[{\cimp\mbox{-int}}]{A\cimp B}
{\deduce{\bot}{\deduce{}{\deduce{\Pi}{\deduce{}{[A,\neg B]}}}}} 
\]
then becomes 
\begin{prooftree}
\AxiomC{$[\cdot ;A]$}
\AxiomC{$\Gamma$}
\noLine
\BinaryInfC{$\Pi$}
\noLine
\UnaryInfC{$\Delta ,B;\cdot$}
\RightLabel{$\to_{c}$-int}
\UnaryInfC{$\Delta ;A \to_{c} B$}
\end{prooftree}
For the elimination rule, also the negated formula in the premise become classical, this time with empty stoup.

\begin{prooftree}
\AxiomC{$\Delta_1; A \to_{c} B$}
\AxiomC{$\Gamma_2$}
\noLine
\UnaryInfC{$\Pi_2$}
\noLine
\UnaryInfC{$\Delta_{2};A$}
\AxiomC{$[\cdot;B]$}
\AxiomC{$\Gamma_3$}
\noLine
\BinaryInfC{$\Pi_3$}
\noLine
\UnaryInfC{$\Delta_{3} ;\cdot$}
\RightLabel{$\to_{c}$-elim}
\TrinaryInfC{$\Delta_{1},\Delta_{2},\Delta_{3};\cdot$}
\end{prooftree}

\subsubsection*{Disjunction.}
The same idea applies to disjunction, where negated assumptions become part of the classical context, while positive assumptions are kept as {\em stoups}
\begin{prooftree}
\AxiomC{$\Delta ,A,B; \cdot $}
\RightLabel{$\vee_{c}$-int}
\UnaryInfC{$\Delta ;A \vee_{c} B$}
\end{prooftree}

\begin{prooftree}
\AxiomC{$\Delta_1;A \vee_{c} B$}
\AxiomC{$[\cdot ;A]$}
\AxiomC{$\Gamma_{2}$}
\noLine
\BinaryInfC{$\Pi_2$}
\noLine
\UnaryInfC{$\Delta_{2} ;\cdot$}
\AxiomC{$[\cdot ;B]$}
\AxiomC{$\Gamma_{3}$}
\noLine
\BinaryInfC{$\Pi_3$}
\noLine
\UnaryInfC{$\Delta_{3} ;\cdot$}
\RightLabel{$\vee_{c}$-elim}
\TrinaryInfC{$\Delta_{1},\Delta_{2},\Delta_{3} ; \cdot$}
\end{prooftree}

\subsection{Neutral operators}
\subsubsection*{Negation.} Since negation can be defined in classical/intuitionistic logic as ``implies bottom'', the rules for $\neg$ can be derived from the ones for implication with an empty {\em stoup}.  

\begin{prooftree}
\AxiomC{$[\cdot ;A]$}
\AxiomC{$\Gamma$}
\noLine
\BinaryInfC{$\Pi$}
\noLine
\UnaryInfC{$\Delta ; \cdot$}
\RightLabel{$\neg$-int}
\UnaryInfC{$\Delta ;\neg A$}
\end{prooftree}

\begin{prooftree}
\AxiomC{$\Delta_{1};A$}
\AxiomC{$\Delta_{2};\neg A$}
\RightLabel{$\neg$-elim}
\BinaryInfC{$\Delta_{1}, \Delta_{2}; \cdot$}
\end{prooftree}

\subsubsection*{Conjunction.} Since our ecumenical system is essentially intuitionistic (in the terms of Proposition~\ref{prop:ev}), all active formulas are placed in the {\em stoup}. Hence we will adopt the multiplicative version of Girard's rule for positive conjuncts.

\begin{prooftree}
\AxiomC{$\Delta_{1};A$}
\AxiomC{$\Delta_{2};B$}
\RightLabel{$\wedge$-int}
\BinaryInfC{$\Delta_{1},\Delta_{2};A \wedge B$}
\end{prooftree}

\begin{prooftree}
\AxiomC{$\Delta ;A_1 \wedge A_2$}
\RightLabel{$\wedge$-elim$_j$}
\UnaryInfC{$\Delta ;A_j$}
\end{prooftree}

\subsection{Hypothesis formation and dereliction}
The hypothesis formation is the usual one and, as in $\LC$, dereliction is needed for guaranteeing  the completeness of the system, since ecumenical active formulas are always placed in the {\em stoup}.

\noindent
\textit{Hypothesis formation}
\begin{prooftree}
\AxiomC{$\cdot;A$}
\end{prooftree}

\textit{Dereliction}
\begin{prooftree}
\AxiomC{$\Delta ;A$}
\RightLabel{der}
\UnaryInfC{$\Delta ,A;\cdot $}
\end{prooftree}

\subsection{Structural rules}
Finally, on choosing the multiplicative version of rules, we need structural rules acting in the classical context, so to transform multisets into sets. As usual in intuitionistic systems, weakening is also allowed in the {\em stoup}.

\textit{Weakening}
\begin{prooftree}
\AxiomC{$\Delta ; \cdot$}
\RightLabel{$W_i$}
\UnaryInfC{$\Delta  ; A$}
\DisplayProof
\qquad \qquad
\AxiomC{$\Delta ; C$}
\RightLabel{$W_c$}
\UnaryInfC{$\Delta , A ; C$}
\end{prooftree}

\textit{Contraction}
\begin{prooftree}
\AxiomC{$\Delta , A,A ; C$}
\RightLabel{$C_c$}
\UnaryInfC{$\Delta , A ; C$}
\end{prooftree}

Derivations are then  inductively defined in the usual way. 
\begin{definition}\label{def:der}
We say that the \textit{stp-c} $\Delta ;\Sigma$ is {\em derivable} from a set $\Gamma$ of \textit{stp-cs}  in $\mathcal{LE}_{p}$ (denoted by $\Gamma \vL \Delta ;\Sigma $) if and only if there is a derivation of $\Delta ;\Sigma$ from $\Gamma$. A formula $A$ is a {\em theorem} of $\mathcal{LE}_{p}$ if and only if $\vL \cdot ; A$. 
\end{definition}
As usual, we also may add indices in derivations, for relating a discharged assumption with a specific rule application.

\section{Examples}
We present below the proofs of some classical tautologies in $\mathcal{LE}_{p}$. 
\begin{enumerate}
\item \textit{Peirce's Law}
\begin{prooftree}
\AxiomC{$[\cdot ; ((A \to_{c} B) \to_{c} A)]^{3}$}
\AxiomC{$[\cdot ;A]^{1}$}
\RightLabel{der}
\UnaryInfC{$A;\cdot$}
\RightLabel{$W_{c}$}
\UnaryInfC{$A,B;\cdot$}
\RightLabel{$\to_{c}$-int}
\LeftLabel{1}
\UnaryInfC{$A;(A \to_{c} B)$}
\AxiomC{$[\cdot ;A]^{2}$}
\RightLabel{der}
\UnaryInfC{$A;\cdot$}
\RightLabel{$\to_{c}$-elim}
\LeftLabel{2}
\TrinaryInfC{$A,A; \cdot $}
\RightLabel{$C_{c}$}
\UnaryInfC{$A; \cdot $}
\LeftLabel{3}
\RightLabel{$\to_{c}$-int}
\UnaryInfC{$\cdot ; (((A \to_{c} B) \to_{c} A) \to_{c} A)$}
\end{prooftree}
Observe that the only difference of the proof above w.r.t. $Alt$ or $Restart$ systems lies in the use of structural rules in the classical context.

More interestingly, note that any sequent of the form  $(((A \to_{j} B) \to_{k} A) \to_{c} A)$ with $j,k\in\{i,c\}$ is provable in  $\mathcal{LE}_{p}$. That is, provability is maintained if the outermost implication is classical.
\item \textit{Excluded-middle}
\begin{prooftree}
\AxiomC{$[\cdot ; A]^{1}$}
\RightLabel{der}
\UnaryInfC{$A; \cdot $}
\RightLabel{$\neg$-int}
\LeftLabel{1}
\UnaryInfC{$ A ; \neg A$}
\RightLabel{der}
\UnaryInfC{$ A , \neg A ; \cdot$}
\RightLabel{$\vee_{c}$-int}
\UnaryInfC{$ \cdot ; (A \vee_{c} \neg A) $}
\end{prooftree}
Of course, $ \cdot ;A \vee_{i} \neg A$ is not a theorem in $\mathcal{LE}_{p}$.
\item \textit{Dummett's linearity axiom}
\begin{prooftree}
\AxiomC{$[\cdot ;A]^{1}$}
\RightLabel{der}
\UnaryInfC{$A; \cdot$}
\RightLabel{$W_{c}$}
\UnaryInfC{$A, B;\cdot $}
\RightLabel{$\to_{c}$-int}
\LeftLabel{$1$}
\UnaryInfC{$A;(A \to_{c} B)$}
\RightLabel{der}
\UnaryInfC{$A, (A \to_{c} B);\cdot $}
\RightLabel{$\to_{c}$-int}
\UnaryInfC{$(A \to_{c} B); (B \to_{c} A)$}
\RightLabel{der}
\UnaryInfC{$(A \to_{c} B), (B \to_{c} A)); \cdot$}
\RightLabel{$\vee_{c}$-int}
\UnaryInfC{$\cdot ; ((A \to_{c} B) \vee_{c}  (B \to_{c} A))$}
\end{prooftree}
This is also an interesting case, where any sequent of the form  $ ((A \to_{j} B) \vee_{c}  (B \to_{k} A))$ with $j,k\in\{i,c\}$ is provable in  $\mathcal{LE}_{p}$. That is, provability is maintained if the outermost conjunction is classical.
\end{enumerate}

\section{Systems equivalence}
In the following, we will show that $\mathcal{LE}_{p}$ is correct and complete w.r.t. $\mathcal{NE}_p$. We will use the following extra notation:
\begin{itemize}
\item[-] Given a multiset $\Delta$ of 
formulas,  we denote by $\neg \Delta $ the multiset formed by the negation of each formula in $\Delta$.
\item[-] If $\Gamma$ is a set of formulas, we denote by $\Gamma \vN A$ the fact that the formula $A$ depends on the set $\Gamma$ of assumptions  in $\mathcal{NE}_{p}$.
\item[-] If $\Gamma=\bigcup\{\cdot;A_i\}_{0\leq i\leq n}$ is a multiset of {\em stoup} with empty contexts  in $\mathcal{LE}_{p}$,  we will abuse the notation and also represent by $\Gamma$ the underlying set of formulas in these {\em stp-c}, that is, $\Gamma=\bigcup\{A_i\}_{0\leq i\leq n}$ in $\mathcal{NE}_{p}$.
\end{itemize}
\begin{theorem}
Let $\Gamma=\bigcup\{\cdot;A_i\}_{0\leq i\leq n}$ be a set of hypothesis. Then 
$\Gamma  \vL \Delta;\Sigma$ iff $\Gamma , \neg \Delta\vN \Sigma$. In case $\Sigma$ is empty, we have $\Gamma , \neg \Delta\vN \bot$. 
\end{theorem}
\begin{proof}
By induction on the length of derivations in $\mathcal{LE}_{p}$ and $\mathcal{NE}_{p}$. The only interesting cases are the ones involving classical connectives and structural rules. 
\begin{itemize}
\item Case $\cvee$-int. Suppose that we have the following derivation  in $\mathcal{LE}_p$: 
\begin{prooftree}
\AxiomC{$\Gamma$}
\noLine
\UnaryInfC{$\Pi$}
\noLine
\UnaryInfC{$\Delta,A,B;\cdot$}
\RightLabel{$\cvee$-int}
\UnaryInfC{$\Delta;A\cvee B $}
\end{prooftree}
By the inductive hypothesis, $\Gamma , \neg A,\neg B,\neg \Delta  \vN\bot$. We can then take the desired derivation to be:
\begin{prooftree}
\AxiomC{$\Gamma\quad \neg \Delta\quad  [\neg A]^{n}\quad  [\neg B]^{m}$}
\noLine
\UnaryInfC{$\Pi'$}
\noLine
\UnaryInfC{$\bot$}
\LeftLabel{$n,m$}
\RightLabel{$\cvee$-int}
\UnaryInfC{$A\cvee B $}
\end{prooftree}

On the other hand, suppose that $ \Gamma , \neg \Delta \vN A\cvee B$  with proof
\begin{prooftree}
\AxiomC{$\Gamma\quad \neg \Delta\quad  [\neg A] \quad [\neg B]$}
\noLine
\UnaryInfC{$\Pi$}
\noLine
\UnaryInfC{$\bot$}
\RightLabel{$\cvee$-int}
\UnaryInfC{$A\cvee B $}
\end{prooftree}
By the inductive hypothesis, we have $\Gamma \vL \Delta,A,B;\cdot$. We can then take the desired derivation to be:
\begin{prooftree}
\AxiomC{$\Gamma$}
\noLine
\UnaryInfC{$\Pi$}
\noLine
\UnaryInfC{$\Delta,A,B;\cdot$}
\RightLabel{$\cvee$-int}
\UnaryInfC{$\Delta;A\cvee B $}
\end{prooftree}

\item Case $\cvee$-elim.
Suppose that $\Gamma_{1}, \Gamma_{2}, \Gamma_{3} \vL\Delta_{1},\Delta_{2},\Delta_{3};\cdot$  with proof 
\begin{prooftree}
\AxiomC{$\Gamma_{1}$}
\noLine
\UnaryInfC{$\Pi_1$}
\noLine
\UnaryInfC{$\Delta_1; A \cvee B$}
\AxiomC{$\Gamma_{2}$}
\AxiomC{$[\cdot;A]$}
\noLine
\BinaryInfC{$\Pi_2$}
\noLine
\UnaryInfC{$\Delta_{2};\cdot$}
\AxiomC{$\Gamma_{3}$}
\AxiomC{$[\cdot;B]$}
\noLine
\BinaryInfC{$\Pi_3$}
\noLine
\UnaryInfC{$\Delta_{3} ;\cdot$}
\RightLabel{$\cvee$-elim}
\TrinaryInfC{$\Delta_{1},\Delta_{2},\Delta_{3};\cdot$}
\end{prooftree}
By the inductive hypothesis we have: $\Gamma_{1}, \neg \Delta_1 \vN A \cvee B$, $\Gamma_{2}, \neg \Delta_2, A\vN \bot$ and 
$\Gamma_{3}, \neg \Delta_3,B\vN \bot$. We can the obtain the desired derivation as follows:
\begin{prooftree}
\AxiomC{$\Gamma_{1}$}
\AxiomC{$\neg\Delta_1$}
\noLine
\BinaryInfC{$\Pi_1'$}
\noLine
\UnaryInfC{$A \cvee B$}
\AxiomC{$\Gamma_{2}$}
\AxiomC{$\neg \Delta_2$}
\AxiomC{$[A]^{n}$}
\noLine
\TrinaryInfC{$\Pi_2'$}
\noLine
\UnaryInfC{$\bot$}
\LeftLabel{$n$}
\RightLabel{$\neg$-int}
\UnaryInfC{$\neg A$}
\AxiomC{$\Gamma_ {3}$}
\AxiomC{$\neg\Delta_3$}
\AxiomC{$[B]^{m}$}
\noLine
\TrinaryInfC{$\Pi_3'$}
\noLine
\UnaryInfC{$\bot$}
\LeftLabel{$m$}
\RightLabel{$\neg$-int}
\UnaryInfC{$\neg B$}
\RightLabel{$\cvee$-elim}
\TrinaryInfC{$\bot$}
\end{prooftree}
On the other hand, suppose that $\Gamma_{1}, \Gamma_{2}, \Gamma_{3}, \neg [\Delta_1],\neg [\Delta_2],\neg [\Delta_3]\vN \bot$  with the following derivation:
\begin{prooftree}
\AxiomC{$\Gamma_{1}$}
\AxiomC{$\neg \Delta_1$}
\noLine
\BinaryInfC{$\Pi_1$}
\noLine
\UnaryInfC{$A \cvee B$}
\AxiomC{$\Gamma_{2}$}
\AxiomC{$\neg \Delta_2$}
\noLine
\BinaryInfC{$\Pi_2$}
\noLine
\UnaryInfC{$\neg A$}
\AxiomC{$\Gamma_{3}$}
\AxiomC{$\neg \Delta_3$}
\noLine
\BinaryInfC{$\Pi_3$}
\noLine
\UnaryInfC{$\neg B$}
\RightLabel{$\cvee$-elim}
\TrinaryInfC{$\bot$}
\end{prooftree}
By the  inductive hypothesis: $\Gamma_{1}, \vL \Delta_1;A\cvee B$, $\Gamma_{2}, \vL \Delta_2;\neg A$ and $\Gamma_{3} \vL \Delta_3;\neg B$. We can then construct our desired derivation as:
\begin{prooftree}
\AxiomC{$\Gamma_{1}$}
\noLine
\UnaryInfC{$\Pi_1'$}
\noLine
\UnaryInfC{$\Delta_1; A \cvee B$}
\AxiomC{$\Gamma_{2}$}
\noLine
\UnaryInfC{$\Pi_2'$}
\noLine
\UnaryInfC{$\Delta_{2};\neg A$}
\AxiomC{$[\cdot; A]$}
\RightLabel{$\neg$-int}
\BinaryInfC{$\Delta_{2};\cdot$}
\AxiomC{$\Gamma_3$}
\noLine
\UnaryInfC{$\Pi_3'$}
\noLine
\UnaryInfC{$\Delta_{3};\neg B$}
\AxiomC{$[\cdot; B]$}
\RightLabel{$\neg$-int}
\BinaryInfC{$\Delta_{2};\cdot$}
\RightLabel{$\cvee$-elim}
\TrinaryInfC{$\Delta_{1},\Delta_{2},\Delta_{3};\cdot$}
\end{prooftree}

\item Case $\cimp$-int.
Suppose that $\Gamma \vL\Delta;A\cimp B$  with a derivation as:
\begin{prooftree}
\AxiomC{$\Gamma$}
\AxiomC{$[\cdot ;A]$}
\noLine
\BinaryInfC{$\Pi$}
\noLine
\UnaryInfC{$\Delta ,B;\cdot$}
\RightLabel{$\to_{c}$-int}
\UnaryInfC{$\Delta ;A \to_{c} B$}
\end{prooftree}
By the inductive hypothesis, we have $ \Gamma , A,\neg B,\neg \Delta \vN\bot$. We can then construct our desired derivation in $\mathcal{NE}_{p}$ as:

\begin{prooftree}
\AxiomC{$\Gamma$}
\AxiomC{$[A]$}
\AxiomC{$[\neg B]$}
\AxiomC{$\neg\Delta$}
\noLine
\QuaternaryInfC{$\Pi'$}
\noLine
\UnaryInfC{$\bot$}
\RightLabel{$\cimp$-int}
\UnaryInfC{$A\cimp B $}
\end{prooftree}
In the other direction, suppose that $ \Gamma ,\neg \Delta  \vN A\cimp B$ with a derivation as:
\begin{prooftree}
\AxiomC{$\Gamma$}
\AxiomC{$\neg \Delta$}
\AxiomC{$[A]$}
\AxiomC{$[\neg B]$}
\noLine
\QuaternaryInfC{$\Pi$}
\noLine
\UnaryInfC{$\bot$}
\RightLabel{$\cimp$-int}
\UnaryInfC{$A\cimp B $}
\end{prooftree}
By the inductive hypothesis, we have $ \Gamma\cup\{  \cdot ; A \} \vL \Delta ,B;\cdot$. Thus, we can obtain the desired derivation in $\mathcal{LE}_{p}$ as:
\begin{prooftree}
\AxiomC{$\Gamma$}
\AxiomC{$[\cdot;A]$}
\noLine
\BinaryInfC{$\Pi'$}
\noLine
\UnaryInfC{$\Delta,B;\cdot$}
\RightLabel{$\cvee$-elim}
\UnaryInfC{$\Delta;A\cimp B $}
\end{prooftree}
\item Case $\cimp$-elim.
Suppose that $ \Gamma_{1}, \Gamma_{2}, \Gamma_{3}  \vL\Delta_{1},\Delta_{2},\Delta_{3};\cdot$   with a derivation as:
\begin{prooftree}
\AxiomC{$\Gamma_{1}$}
\noLine
\UnaryInfC{$\Pi_1$}
\noLine
\UnaryInfC{$\Delta_1; A \to_{c} B$}
\noLine
\AxiomC{$\Gamma_{2}$}
\noLine
\UnaryInfC{$\Pi_2$}
\noLine
\UnaryInfC{$\Delta_{2};A$}
\AxiomC{$\Gamma_{3}$}
\AxiomC{$[\cdot;B]$}
\noLine
\BinaryInfC{$\Pi_3$}
\noLine
\UnaryInfC{$\Delta_{3} ;\cdot$}
\RightLabel{$\to_{c}$-elim}
\TrinaryInfC{$\Delta_{1},\Delta_{2},\Delta_{3};\cdot$}
\end{prooftree}
By inductive hypothesis, we have $ \Gamma_{1} ,\neg \Delta_1  \vN A \to_{c} B$, $ \Gamma_{2}, \neg \Delta_2 \vN A$ and 
$ \Gamma_{3}, \neg\Delta_3 ,B \vN \bot$. Then, we can obtain the desired derivation in $\mathcal{NE}_{p}$ as:
\begin{prooftree}
\AxiomC{$\Gamma_{1}$}
\AxiomC{$\neg \Delta_1$}
\noLine
\BinaryInfC{$\Pi_1'$}
\noLine
\UnaryInfC{$A \to_{c} B$}
\AxiomC{$\Gamma_{2}$}
\AxiomC{$\neg \Delta_2$}
\noLine
\BinaryInfC{$\Pi_2'$}
\noLine
\UnaryInfC{$A$}
\AxiomC{$\Gamma_{3}$}
\AxiomC{$\neg \Delta_3 \quad[B]$}
\noLine
\BinaryInfC{$\Pi_3'$}
\noLine
\UnaryInfC{$\bot$}
\RightLabel{$\neg$-int}
\UnaryInfC{$\neg B$}
\RightLabel{$\to_{c}$-elim}
\TrinaryInfC{$\bot$}
\end{prooftree}
In the other direction, suppose that $ \Gamma_{1}, \Gamma_{2}, \Gamma_{3}, \neg \Delta_1,\neg \Delta_2,\neg \Delta_3  \vN \bot$   with a derivation as:
\begin{prooftree}
\AxiomC{$\Gamma_{1}$}
\AxiomC{$\neg \Delta_1$}
\noLine
\BinaryInfC{$\Pi_1$}
\noLine
\UnaryInfC{$A \cimp B$}
\AxiomC{$\Gamma_{2}$}
\AxiomC{$\neg \Delta_2$}
\noLine
\BinaryInfC{$\Pi_2$}
\noLine
\UnaryInfC{$A$}
\AxiomC{$\Gamma_{3}$}
\AxiomC{$\neg \Delta_3$}
\noLine
\BinaryInfC{$\Pi_3$}
\noLine
\UnaryInfC{$\neg B$}
\RightLabel{$\cimp$-elim}
\TrinaryInfC{$\bot$}
\end{prooftree}
By the inductive hypothesis, we have that $\Gamma_{1} \vL \Delta_1;A\cimp B$, $\Gamma_{2} \vL \Delta_2; A$ and $\Gamma_{3} \vL \Delta_3;\neg B$. We can then  obtain the desired derivation in $\mathcal{LE}_{p}$ as:

\begin{prooftree}
\AxiomC{$\Gamma_{1}$}
\noLine
\UnaryInfC{$\Pi_1'$}
\noLine
\UnaryInfC{$\Delta_1; A \cimp B$}
\AxiomC{$\Gamma_{2}$}
\noLine
\UnaryInfC{$\Pi_2'$}
\noLine
\UnaryInfC{$\Delta_{2};A$}
\AxiomC{$\Gamma_{3}$}
\noLine
\UnaryInfC{$\Pi_3'$}
\noLine
\UnaryInfC{$\Delta_{3};\neg B$}
\AxiomC{$[\cdot; B]$}
\RightLabel{$\neg$-int}
\BinaryInfC{$\Delta_{2};\cdot$}
\RightLabel{$\cvee$-elim}
\TrinaryInfC{$\Delta_{1},\Delta_{2},\Delta_{3};\cdot$}
\end{prooftree}

\item Case der. Consider the derivation
\begin{prooftree}
\AxiomC{$\Gamma$}
\noLine
\UnaryInfC{$\Pi$}
\noLine
\UnaryInfC{$\Delta ;A$}
\RightLabel{der}
\UnaryInfC{$\Delta ,A;\cdot $}
\end{prooftree}
By the  inductive hypothesis, we have $ \Gamma ,\neg \Delta \vN A$. We can then obtain the desired derivation in $\mathcal{NE}_{p}$ as:

\begin{prooftree}
\AxiomC{$\Gamma$}
\AxiomC{$\neg\Delta$}
\noLine
\BinaryInfC{$\Pi'$}
\noLine
\UnaryInfC{$A$}
\AxiomC{$[\neg A]$}
\RightLabel{$\neg$-int}
\BinaryInfC{$\bot $}
\end{prooftree}
The cases of the other structural rules are trivial.
\end{itemize}
\end{proof}

\section{Normalization}
We will now describe how normalization works in the natural deduction with {\em stoup} setting.
The idea follows the usual one for natural deduction systems: show how to {\em compose} derivations, so to eliminate detours. The presence of {\em stoups}, however, adds an extra case analysis, since the composition may occur in the {\em stoup} or in the classical context. Both processes will be carefully described in what follows.
\subsection{Composition}
Before we define reductions and prove the normalization theorem for $\mathcal{LE}_{p}$, we must guarantee that the process of \textit{composition} of derivations is preserved in $\mathcal{LE}_{p}$. As quickly mentioned in Sec.~\ref{sec:sc}, sequent systems with {\em stoup} usually allow for two types of cut: a cut where the left cut-formula is in the \textit{stoup} and a cut where the left cut-formula is in the classical region. These two types of cut will correspond to two modes of composition: a composition that occurs in the \textit{stoup} and a composition that occurs in the classical context. We will detail these two forms of compositions below.
\begin{enumerate}
\item Composition in the \textit{stoup}.
\begin{theorem}\label{thm:comp-stoup}
Let $\Pi_{1}$ be a derivation of $\Gamma_{1} \vL \Delta_{1} ; A$ and $\Pi_{2}$ be a derivation of $\Gamma_{2} \cup \{\cdot ;A\} \vL \Delta_{2} ;B$. Then, the result of replacing the assumption $\cdot ;A$ in $\Pi_{2}$ by the derivation $\Pi_{1}$ is a derivation $\Pi$ of $\Gamma_{1} , \Gamma_{2} \vL \Delta_{1}  , \Delta_{2}; B$. 
\end{theorem}
\begin{proof}
By induction on the length of $\Pi_{2}$. We will examine two cases, the other cases being treated in a similar way.\\
\begin{enumerate}
\item The last rule applied in $\Pi_{2}$ is $\vee_{c}$-elim:
\begin{prooftree}
\AxiomC{$\Gamma_{2}^{1}$}
\AxiomC{$\cdot ;A$}
\noLine
\BinaryInfC{$\Pi_{2}^1$}
\noLine
\UnaryInfC{$\Delta_{2}^{1};C \vee_{c} D$}
\AxiomC{$[\cdot ;C]$}
\AxiomC{$\cdot ;A$}
\AxiomC{$\Gamma_{2}^{2}$}
\noLine
\TrinaryInfC{$\Pi_{2}^2$}
\noLine
\UnaryInfC{$\Delta_{2}^{2} ;\cdot$}
\AxiomC{$[\cdot ;D]$}
\AxiomC{$\cdot ;A$}
\AxiomC{$\Gamma_{2}^{3}$}
\noLine
\TrinaryInfC{$\Pi_{2}^3$}
\noLine
\UnaryInfC{$\Delta_{2}^{3} ;\cdot$}
\RightLabel{$\vee_{c}$-elim}
\TrinaryInfC{$\Delta_{2} ; \cdot$}
\end{prooftree}
where $\Delta_2^{i}$ indicates a partition of $\Delta_2$, the same with $\Gamma_2$. 
By the induction hypothesis we obtain the following derivations:
\begin{prooftree}
\AxiomC{$\Gamma_{1}$}
\AxiomC{$\Gamma_{2}^{1}$}
\noLine
\BinaryInfC{$\Pi_{1}^{*}$}
\noLine
\UnaryInfC{$\Delta_{1} , \Delta_{2}^{1} ; (C \vee_{c} D)$}
\end{prooftree}

and

\begin{prooftree}
\AxiomC{$\Gamma_{1}$}
\AxiomC{$[\cdot ; C]$}
\AxiomC{$\Gamma_{2}^{2}$}
\noLine
\TrinaryInfC{$\Pi_{2}^{*}$}
\noLine
\UnaryInfC{$\Delta_{1} , \Delta_{2}^{2} ; \cdot$}
\end{prooftree}

and

\begin{prooftree}
\AxiomC{$\Gamma_{1}$}
\AxiomC{$[\cdot ; D]$}
\AxiomC{$\Gamma_{2}^{3}$}
\noLine
\TrinaryInfC{$\Pi_{3}^{*}$}
\noLine
\UnaryInfC{$\Delta_{1} , \Delta_{2}^{3} ; \cdot$}
\end{prooftree}

The resulting derivation $\Pi$ is:
\begin{prooftree}
\AxiomC{$\Gamma_{1}$}
\AxiomC{$\Gamma_{2}^{1}$}
\noLine
\BinaryInfC{$\Pi_{1}^{*}$}
\noLine
\UnaryInfC{$\Delta_{1} , \Delta_{2}^{1} ; (C \vee_{c} D)$}
\AxiomC{$\Gamma_{1}$}
\AxiomC{$[\cdot ; C]$}
\AxiomC{$\Gamma_{2}^{2}$}
\noLine
\TrinaryInfC{$\Pi_{2}^{*}$}
\noLine
\UnaryInfC{$\Delta_{1} , \Delta_{2}^{2} ; \cdot$}
\AxiomC{$\Gamma_{1}$}
\AxiomC{$[\cdot ; D]$}
\AxiomC{$\Gamma_{2}^{3}$}
\noLine
\TrinaryInfC{$\Pi_{3}^{*}$}
\noLine
\UnaryInfC{$\Delta_{1} , \Delta_{2}^{3} ; \cdot$}
\RightLabel{$\vee_{c}$-elim}
\TrinaryInfC{$\Delta_{1} , \Delta_{1} , \Delta_{1} , \Delta_{2} ; \cdot$}
\doubleLine
\RightLabel{$C_c$}
\UnaryInfC{$\Delta_{1} , \Delta_{2} ; \cdot$}
\end{prooftree}

\noindent Where the \textit{double line} indicates several applications of $C_c$.

\item The last rule applied in $\Pi_{2}$ is dereliction:

\begin{prooftree}
\AxiomC{$\Gamma_{2}$}
\AxiomC{$\cdot ;A$}
\noLine
\BinaryInfC{$\Pi_{2}$}
\noLine
\UnaryInfC{$\Delta_{2} ;B$}
\RightLabel{der}
\UnaryInfC{$\Delta_{2} ,B;\cdot $}
\end{prooftree}

By the induction hypothesis we obtain the following derivation:
\begin{prooftree}
\AxiomC{$\Gamma_{1}$}
\AxiomC{$\Gamma_{2}$}
\noLine
\BinaryInfC{$\Pi^{*}$}
\noLine
\UnaryInfC{$\Delta_{1} , \Delta_{2} ; B$}
\end{prooftree}

The resulting derivation $\Pi$ is then:

\begin{prooftree}
\AxiomC{$\Gamma_{1}$}
\AxiomC{$\Gamma_{2}$}
\noLine
\BinaryInfC{$\Pi^{*}$}
\noLine
\UnaryInfC{$\Delta_{1} , \Delta_{2} ; B$}
\RightLabel{der}
\UnaryInfC{$\Delta_{1} , \Delta_{2} , B; \cdot $}
\end{prooftree}

\end{enumerate}
\end{proof}
\item Composition in the \textit{context}.
\begin{theorem}
Let $\Pi_{1}$ be a derivation of $\Gamma_{1} \vL \Delta_{1};C$, where $A \in \Delta_{1}$, and $\Pi_{2}$ be a derivation of $ \Gamma_{2}\cup\{\cdot; A\} \vL \Delta_{2} ;B$. Then, the result of replacing the assumption $\cdot ;A$ in $\Pi_{2}$ by the derivation $\Pi_{1}$ is a derivation $\Pi$ of $\Gamma_{1} , \Gamma_{2} \vL \Delta_{1}^{*}  , \Delta_{2}, B;C$, where  $\Delta_{1}^{*}$ is obtained from $\Delta_{1}$ by means of the elimination of the occurrences of $A$.
\end{theorem}
\begin{proof}
The proof is by straightforward induction on the length of $\Pi_{1}$. The only exception is dereliction, which uses Theorem~\ref{thm:comp-stoup}.\\ \\
Let the last rule applied in $\Pi_{1}$ be \textit{dereliction}:

\begin{prooftree}
\AxiomC{$\Gamma_1$}
\noLine
\UnaryInfC{$\Pi_{1}$}
\noLine
\UnaryInfC{$\Delta_{1} ; A$}
\RightLabel{der}
\UnaryInfC{$\Delta_{1} , A; \cdot $}
\end{prooftree}
\end{proof}

By the induction hypothesis we can obtain a derivation $\Pi'$  of  $\Gamma_{1}, \Gamma_{2} \vL \Delta_{1}^{*} , \Delta_{2}, B;A$. By Theorem 2 we can obtain a derivation $\Pi''$ of 
$\Gamma_{1}, \Gamma_{2} \vL \Delta_{1}^{*} , \Delta_{2}, \Delta_{2}, B;B$. We can now take the derivation $\Pi$ to be:
\begin{prooftree}
\AxiomC{$\Gamma_{1}$}
\AxiomC{$\Gamma_{2}$}
\noLine
\BinaryInfC{$\Pi''$}
\noLine
\UnaryInfC{$\Delta_{1}^{*} , \Delta_{2}, \Delta_{2}, B;B$}
\RightLabel{$der$}
\UnaryInfC{$\Delta_{1}^{*} , \Delta_{2}, \Delta_{2}, B, B; \cdot$}
\doubleLine
\RightLabel{$C_{c}$}
\UnaryInfC{$\Delta_{1}^{*} , \Delta_{2}, B; \cdot$}
\end{prooftree}

\end{enumerate}

In what follows, we shall use the following notation to indicate composition in the \textit{stoup}  and composition  in the \textit{classical context}.
\begin{itemize}
\item Composition in the \textit{stoup}:
\begin{prooftree}
\AxiomC{$\Gamma_{1}$}
\noLine
\UnaryInfC{$\Pi_{1}$}
\noLine
\UnaryInfC{$\Delta_{1} ; [\frac{A}{\cdot ;A}]$}
\AxiomC{$\Gamma_{2}$}
\noLine
\BinaryInfC{$\Pi_{2}$}
\noLine
\UnaryInfC{$\Delta_{1} , \Delta_{2} ;B$}
\end{prooftree}

\item Composition in the \textit{classical context}:
\begin{prooftree}
\AxiomC{$\Gamma_{1}$}
\noLine
\UnaryInfC{$\Pi_{1}$}
\noLine
\UnaryInfC{$\Delta_{1} , [\frac{A}{\cdot ;A}]; B$}
\AxiomC{$\Gamma_{2}$}
\noLine
\BinaryInfC{$\Pi_{2}$}
\noLine
\UnaryInfC{$\Delta_{1} , \Delta_{2} ;B$}
\end{prooftree}

\item A more concise notation for composition in general: $\Pi_{1}/[\frac{A}{\cdot ;A}]/\Pi_{2}$.

\end{itemize}

\subsection{Reductions}
Derivations in $\mathcal{LE}_{p}$ may contain \textit{detours}. These detours are of two types: we may introduce a formula by an application of an introduction rule to immediately use it as major premiss of an application of an elimination rule; or we may introduce a formula by an application of an introduction rule and use it as major premiss of an application of an elimination rule after several applications of $\vee_{i}$-elim. The \textit{reductions} defined in this section are intended, as usual, to eliminate \textit{detours} that may occur in a derivation.
\begin{definition}
A {\em segment} in a derivation $\Pi$ is a sequence $A_{1},. . . , A_{n}$  of consecutive formulas  in a thread in $\Pi$ such that:
\begin{itemize}
\item $A_{1}$ is not in the stoup of the consequence of an application of $\vee_{i}$-elim or of an application of $C_{c}$;
\item $A_{j}$, for $j < n$, is in the stoup of the minor premiss of an application of $\vee_{i}$-elim or of an application of $C_{c}$; and
\item $A_{n}$ is not in the stoup of  the minor premiss of an application of $\vee_{i}$-elim or of an application of $C_{c}$.
\end{itemize}
\end{definition}
We note  the presence of contraction in the last definition. The idea is that contractions move down on reductions, just like in a sequent calculus' cut-elimination process.

\begin{definition}
A segment that begins with with the consequence of an application of an introduction rule or $W_{i}$ and ends with an application of an elimination rule is called a {\em maximal segment}. A maximal segment of length $1$ is called a {\em maximum formula}.
\end{definition}

\begin{definition}
Let $\Pi$ be a derivation in $\mathcal{LE}_{p}$. The {\em degree of $\Pi$}, $d[\Pi ]$, is defined as $\max\{d[A]:$ A is the end-formula of maximal segment in $\Pi \}$, where $d[A]$ is the {\em weight} of the formula $A$, defined inductively by
$$
\begin{array}{lcl}
d[\bot]=d[p] &=& 0 \quad p\mbox{ atomic.}\\
d[A \circ B] &=& d[A] + d[B] + 1 \mbox{ for } \circ\in\{\to_{i,c},\vee_{i,c},\wedge\}\\
d[\neg A] &=& d[A] +1.
\end{array}
$$
\end{definition}

\begin{definition}
A derivation $\Pi$ is called {\em normal} if and only if $d[\Pi ] = 0$.
\end{definition}
We will present next all the reduction steps in $\mathcal{LE}_{p}$ that will be used in the elimination of maximal segments.
\begin{enumerate}
\item $\wedge$-reduction:\\

The derivation
\begin{prooftree}
\AxiomC{$\Gamma_{1}$}
\noLine
\UnaryInfC{$\Delta_{1} ;A_{1}$}
\AxiomC{$\Gamma_{2}$}
\noLine
\UnaryInfC{$\Delta_{2} ;A_{2}$}
\RightLabel{$\wedge$-int}
\BinaryInfC{$\Delta_{1} , \Delta_{2} ; (A_{1} \wedge A_{2})$}
\RightLabel{$\wedge_j$-elim}
\UnaryInfC{$\Delta_{1} , \Delta_{2} ;A_{j}$}
\end{prooftree}

Reduces to
\begin{prooftree}
\AxiomC{$\Gamma_{j}$}
\noLine
\UnaryInfC{$\Delta_{j} ;A_{j}$}
\RightLabel{$W_c$}
\doubleLine
\UnaryInfC{$\Delta_{1} , \Delta_{2} ; A_{j}$}
\end{prooftree}

\item $\to_{i}$-reduction:\\

The derivation
\begin{prooftree}
\AxiomC{$\Gamma_{1}$}
\noLine
\UnaryInfC{$\Pi_{1}$}
\noLine
\UnaryInfC{$\Delta_{1} ; A$}
\AxiomC{$[\cdot ;A]$}
\AxiomC{$\Gamma_{2}$}
\noLine
\BinaryInfC{$\Pi_{2}$}
\noLine
\UnaryInfC{$\Delta_{2} ; B$}
\RightLabel{$\iimp$-int}
\UnaryInfC{$\Delta_{2} ; (A \to_{i} B)$}
\RightLabel{$\iimp$-elim}
\BinaryInfC{$\Delta_{1} , \Delta_{2} ;B$}
\end{prooftree}

Reduces to:
\begin{prooftree}
\AxiomC{$\Gamma_{1}$}
\noLine
\UnaryInfC{$\Pi_{1}$}
\noLine
\UnaryInfC{$\Delta_{1} ; [\frac{A}{\cdot ;A}]$}
\AxiomC{$\Gamma_{2}$}
\noLine
\BinaryInfC{$\Pi_{2}$}
\noLine
\UnaryInfC{$\Delta_{1} , \Delta_{2} ; B$}
\end{prooftree}
Observe that the case for negation is analogous.

\item $\vee_{i}$-reduction:\\

The derivation

\begin{prooftree}
\AxiomC{$\Gamma_{1}$}
\noLine
\UnaryInfC{$\Pi_{1}$}
\noLine
\UnaryInfC{$\Delta_{1} ;A_{j}$}
\RightLabel{$\ivee$-int}
\UnaryInfC{$\Delta_{1} ; (A_{1} \vee_{i} A_{2})$}
\AxiomC{$\Gamma_{2}$}
\AxiomC{$[\cdot ;A_{1}]$}
\noLine
\BinaryInfC{$\Pi_{2}$}
\noLine
\UnaryInfC{$\Delta_{2} ; B$}
\AxiomC{$\Gamma_{3}$}
\AxiomC{$[\cdot ;A_{2}]$}
\noLine
\BinaryInfC{$\Pi_{3}$}
\noLine
\UnaryInfC{$\Delta_{3} ; B$}
\RightLabel{$\ivee$-elim}
\TrinaryInfC{$\Delta_{1} , \Delta_{2} ,\Delta_{3} ; B$}

\end{prooftree}

Reduces to:
\begin{prooftree}
\AxiomC{$\Gamma_{1}$}
\noLine
\UnaryInfC{$\Pi_{1}$}
\noLine
\UnaryInfC{$\Delta_{1} ; [\frac{A_j}{\cdot ;A_{j}}]$}
\AxiomC{$\Gamma_{j}$}
\noLine
\BinaryInfC{$\Pi_{j}$}
\noLine
\UnaryInfC{$\Delta_{1} , \Delta_{j} ; B$}
\doubleLine
\RightLabel{$W_c$}
\UnaryInfC{$\Delta_{1} , \Delta_{2}, \Delta_{3} ; B$}
\end{prooftree}

\item $\to_{c}$-reduction:\\

The derivation
\begin{prooftree}
\AxiomC{$[\cdot ;A]$}
\AxiomC{$\Gamma_{1}$}
\noLine
\BinaryInfC{$\Pi_{1}$}
\noLine
\UnaryInfC{$\Delta_{2} , B; \cdot$}
\RightLabel{$\cimp$-int}
\UnaryInfC{$\Delta_{2} ; (A \to_{c} B)$}
\AxiomC{$\Gamma_2$}
\noLine
\UnaryInfC{$\Pi_2$}
\noLine
\UnaryInfC{$\Delta_{2};A$}
\AxiomC{$[\cdot;B]$}
\AxiomC{$\Gamma_3$}
\noLine
\BinaryInfC{$\Pi_3$}
\noLine
\UnaryInfC{$\Delta_{3} ;\cdot$}
\RightLabel{$\to_{c}$-elim}
\TrinaryInfC{$\Delta_{1},\Delta_{2},\Delta_{3};\cdot$}
\end{prooftree}

Reduces to:

\begin{prooftree}
\AxiomC{$\Gamma_2$}
\noLine
\UnaryInfC{$\Pi_2$}
\noLine
\UnaryInfC{$\Delta_{2};[\frac{A}{\cdot ;A}]$}
\AxiomC{$\Gamma_{1}$}
\noLine
\BinaryInfC{$\Pi_{1}$}
\noLine
\UnaryInfC{$\Delta_{1} , \Delta_{2} , [\frac{B}{\cdot ;B}]; \cdot$}
\AxiomC{$\Gamma_{3}$}
\noLine
\BinaryInfC{$\Pi_{3}$}
\noLine
\UnaryInfC{$\Delta_{1} , \Delta_{2}^{*} , \Delta_{3} ; \cdot$}
\doubleLine
\RightLabel{$W_{c}$}
\UnaryInfC{$\Delta_{1} , \Delta_{2}, \Delta_{3} ; \cdot$}
\end{prooftree}

\item $\vee_{c}$-reduction:\\

The derivation

\begin{prooftree}
\AxiomC{$\Gamma_{1}$}
\noLine
\UnaryInfC{$\Pi_{1}$}
\noLine
\UnaryInfC{$\Delta_{1} , A,B;\cdot$}
\RightLabel{$\cvee$-int}
\UnaryInfC{$\Delta_{1} ; (A \vee_{c} B)$}
\AxiomC{$\Gamma_{2}$}
\AxiomC{$[\cdot ;A]$}
\noLine
\BinaryInfC{$\Pi_{2}$}
\noLine
\UnaryInfC{$\Delta_{2} ; \cdot$}
\AxiomC{$\Gamma_{3}$}
\AxiomC{$[\cdot ;B]$}
\noLine
\BinaryInfC{$\Pi_{3}$}
\noLine
\UnaryInfC{$\Delta_{3} ; \cdot$}
\RightLabel{$\cvee$-elim}
\TrinaryInfC{$\Delta_{1} , \Delta_{2} ,\Delta_{3} ; \cdot$}

\end{prooftree}

Reduces to:

\begin{prooftree}
\AxiomC{$\Gamma_{1}$}
\noLine
\UnaryInfC{$\Pi_{1}$}
\noLine
\UnaryInfC{$\Delta_{1} , [\frac{A}{\cdot ;A}],B;\cdot$}
\AxiomC{$\Gamma_{2}$}
\noLine
\BinaryInfC{$\Pi_{2}$}
\noLine
\UnaryInfC{$\Delta_{1}^{*} , \Delta_{2} , [\frac{B}{\cdot ;B}];\cdot$}
\AxiomC{$\Gamma_{3}$}
\noLine
\BinaryInfC{$\Pi_{3}$}
\noLine
\UnaryInfC{$\Delta_{1}^{*} , \Delta_{2}^{*} , \Delta_{3};\cdot$}
\doubleLine
\RightLabel{$W_{c}$}
\UnaryInfC{$\Delta_{1}, \Delta_{2}, \Delta_{3};\cdot$}
\end{prooftree}

\item Permutative reductions:\\
\begin{enumerate}
\item The derivation

{\small
\begin{prooftree}
\AxiomC{$\Gamma_{1}$}
\noLine
\UnaryInfC{$\Pi_{1}$}
\noLine
\UnaryInfC{$\Delta_{1} ; (A \vee_{i} B)$}
\AxiomC{$[\cdot ;A]$}
\AxiomC{$\Gamma_{2}$}
\noLine
\BinaryInfC{$\Pi_{2}$}
\noLine
\UnaryInfC{$\Delta_{2} ; C$}
\AxiomC{$[\cdot ;B]$}
\AxiomC{$\Gamma_{3}$}
\noLine
\BinaryInfC{$\Pi_{3}$}
\noLine
\UnaryInfC{$\Delta_{3} ; C$}
\TrinaryInfC{$\Delta_{1} , \Delta_{2} , \Delta_{3} ; C$}
\AxiomC{$\Sigma_{1}$}
\noLine
\UnaryInfC{$\Theta_{1};\Lambda_{1}$}
\AxiomC{$\ldots$}
\AxiomC{$\Sigma_{m}$}
\noLine
\UnaryInfC{$\Theta_{m};\Lambda_{m}$}
\QuaternaryInfC{$\Delta_{1} , \Delta_{2} , \Delta_{3} , \Theta_{1} ,\ldots, \Theta_{m} ; \Lambda$}
\end{prooftree}
}
where $C$ is the major premiss of an elimination rule with minor premisses $\Theta_{1};\Lambda_{1} \ldots \Theta_{m};\Lambda_{m}$ (if any), reduces to \\

\resizebox{\textwidth}{!}{ 
$
\infer={\Delta_1,\Delta_2,\Delta_3,\Theta_1,\ldots,\Theta_m;\Lambda}
{\infer{\Delta_1,\Delta_2,\Delta_3,\Theta_1,\ldots,\Theta_m,\Theta_1,\ldots,\Theta_m;\Lambda}
{\deduce{\Delta_1;A\ivee B}
{\deduce{\vspace{0.1cm}\Pi_1}{\vspace{0.1cm}\Gamma_1}}&
\infer{\Delta_2,\Theta_1,\ldots,\Theta_m;\Lambda}
{\deduce{\Delta_2;C}{\deduce{\vspace{0.1cm}\Pi_2}
{\vspace{0.1cm}[\cdot;A]&\vspace{0.1cm}\Gamma_2}}&
\deduce{\Theta_1;\Lambda_1}{\vspace{0.1cm}\Sigma_1}&
\ldots&
\deduce{\Theta_m;\Lambda_m}{\vspace{0.1cm}\Sigma_m}
}&
\infer{\Delta_3,\Theta_1,\ldots,\Theta_m;\Lambda}
{\deduce{\Delta_3;C}{\deduce{\vspace{0.1cm}\Pi_3}
{\vspace{0.1cm}[\cdot;B]&\vspace{0.1cm}\Gamma_3}}&
\deduce{\Theta_1;\Lambda_1}{\vspace{0.1cm}\Sigma_1}&
\ldots&
\deduce{\Theta_m;\Lambda_m}{\vspace{0.1cm}\Sigma_m}
}}}
$
}

\item The derivation 
{\small
\begin{prooftree}
\AxiomC{$\Gamma_{1}$}
\noLine
\UnaryInfC{$\Pi_{1}$}
\noLine
\UnaryInfC{$\Delta_{1}, A, A;C$}
\RightLabel{$C_{c}$}
\UnaryInfC{$\Delta_{1}, A;C$}
\AxiomC{$\Sigma_{1}$}
\noLine
\UnaryInfC{$\Theta_{1};\Lambda_{1}$}
\AxiomC{$\ldots$}
\AxiomC{$\Sigma_{m}$}
\noLine
\UnaryInfC{$\Theta_{m};\Lambda_{m}$}
\QuaternaryInfC{$\Delta_{1} , \Delta_{2} , \Delta_{3} , \Theta_{1} ,\ldots, \Theta_{m} ; \Lambda$}
\end{prooftree}
}
where $C$ is the major premiss of an elimination rule with minor premisses $\Theta_{1};\Lambda_{1} \ldots \Theta_{m};\Lambda_{m}$ (if any), reduces to \\

{\small
\begin{prooftree}
\AxiomC{$\Gamma_{1}$}
\noLine
\UnaryInfC{$\Pi_{1}$}
\noLine
\UnaryInfC{$\Delta, A, A:C$}
\AxiomC{$\Sigma_{1}$}
\noLine
\UnaryInfC{$\Theta_{1};\Lambda_{1}$}
\AxiomC{$\ldots$}
\AxiomC{$\Sigma_{m}$}
\noLine
\UnaryInfC{$\Theta_{m};\Lambda_{m}$}
\QuaternaryInfC{$\Delta, A, A, \Theta_{1} ,\ldots, \Theta_{m}; \Lambda$}
\RightLabel{$C_{c}$}
\UnaryInfC{$\Delta , A,  \Theta_{1} ,\ldots, \Theta_{m}, A; \Lambda$}
\end{prooftree}
}

\end{enumerate}
\end{enumerate}

\subsection{Normalization}
We shall use Pottinger's \textit{critical derivation} strategy ~\cite{Pottinger76} to prove the normalization theorem for $\mathcal{LE}_{p}$. But before the proof of normalization, we need some definitions and preparatory lemmas that relate reductions and composition to the degree of derivations. The proof of the next lemma is obvious.

\begin{lemma}\label{lemma:junction}
Let $\Pi$ be $\Pi_{1} /[\frac{A}{\cdot ;A}]/\Pi_{2}$, the composition of derivations $\Pi_{1}$ with $\Pi_ {2}$ at junction point $A$. Then, $d[\Pi ] = max\{d[\Pi_{1} ] , d[\Pi_{2} ], d[A]\}$
\end{lemma}

\begin {lemma} \label{lemma:reduction}
If $\Pi$ reduces to $\Pi'$, then $d[\Pi ] \leq d[\Pi']$.
\end{lemma}
\begin{proof}
Directly from the form of the reductions and Lemma~\ref{lemma:junction}.
\end{proof}
\begin{definition}
A derivation $\Pi$ is {\em critical} iff:
\begin{itemize}
\item $\Pi$ ends with an elimination rule $\alpha$;
\item The major premiss $A$ of $\alpha$ is the end of maximal segment;
\item $d[\Pi ]$ = $d[A]$; and
\item For every proper subderivation $\Pi'$ of $\Pi$, $d[\Pi' ]$ $<$ $d[\Pi ]$.
\end{itemize}
\end{definition}

\begin{lemma}\label{lemma:critical}
\textit{(\textbf{Critical Lemma}}): Let $\Pi$ be a critical derivation of $\Gamma\vL \Delta ; \Sigma$  in $\mathcal{LE}_{p}$. Then, $\Pi$ reduces to a derivation $\Pi'$ of $\Gamma'\vL\Delta ; \Sigma$ with $\Gamma ' \subseteq \Gamma$, such that $d[\Pi' ]$ $<$ $d[\Pi ]$.
\end{lemma}
\begin{proof}
By induction on the length of $\Pi$.
\begin{itemize}
\item Case 1: The major premiss of the last rule applied in $\Pi$ is a maximum formula. The result follows directly from the form of the reductions and Lemma~\ref{lemma:reduction}.
\item Case 2: The major premiss of the last rule applied in $\Pi$ is the end formula of maximum segment of length $>$1. There are two sub-cases to be examined:
\begin{enumerate}
\item $\Pi$ is:
{\small
\begin{prooftree}
\AxiomC{$\Gamma_{1}$}
\noLine
\UnaryInfC{$\Pi_{1}$}
\noLine
\UnaryInfC{$\Delta_{1} ; (A \vee_{i} B)$}
\AxiomC{$[\cdot ;A]$}
\AxiomC{$\Gamma_{2}$}
\noLine
\BinaryInfC{$\Pi_{2}$}
\noLine
\UnaryInfC{$\Delta_{2} ; C$}
\AxiomC{$[\cdot ;B]$}
\AxiomC{$\Gamma_{3}$}
\noLine
\BinaryInfC{$\Pi_{3}$}
\noLine
\UnaryInfC{$\Delta_{3} ; C$}
\TrinaryInfC{$\Delta_{1} , \Delta_{2} , \Delta_{3} ; C$}
\AxiomC{$\Sigma_{1}$}
\noLine
\UnaryInfC{$\Theta_{1};\Lambda_{1}$}
\AxiomC{$\ldots$}
\AxiomC{$\Sigma_{m}$}
\noLine
\UnaryInfC{$\Theta_{m};\Lambda_{m}$}
\QuaternaryInfC{$\Delta_{1} , \Delta_{2} , \Delta_{3} , \Theta_{1} ,\ldots, \Theta_{m} ; \Lambda$}
\end{prooftree}
}
By a permutative reduction, $\Pi$ reduces to the following derivation $\Pi^{*}$:\\

\resizebox{\textwidth}{!}{ 
$
\infer={\Delta_1,\Delta_2,\Delta_3,\Theta_1,\ldots,\Theta_m;\Lambda}
{\infer{\Delta_1,\Delta_2,\Delta_3,\Theta_1,\ldots,\Theta_m,\Theta_1,\ldots,\Theta_m;\Lambda}
{\deduce{\Delta_1;A\ivee B}
{\deduce{\vspace{0.1cm}\Pi_1}{\vspace{0.1cm}\Gamma_1}}&
\infer{\Delta_2,\Theta_1,\ldots,\Theta_m;\Lambda}
{\deduce{\Delta_2;C}{\deduce{\vspace{0.1cm}\Pi_2}
{\vspace{0.1cm}[\cdot;A]&\vspace{0.1cm}\Gamma_2}}&
\deduce{\Theta_1;\Lambda_1}{\vspace{0.1cm}\Sigma_1}&
\ldots&
\deduce{\Theta_m;\Lambda_m}{\vspace{0.1cm}\Sigma_m}
}&
\infer{\Delta_3,\Theta_1,\ldots,\Theta_m;\Lambda}
{\deduce{\Delta_3;C}{\deduce{\vspace{0.1cm}\Pi_3}
{\vspace{0.1cm}[\cdot;B]&\vspace{0.1cm}\Gamma_3}}&
\deduce{\Theta_1;\Lambda_1}{\vspace{0.1cm}\Sigma_1}&
\ldots&
\deduce{\Theta_m;\Lambda_m}{\vspace{0.1cm}\Sigma_m}
}}}
$
}
\\ \\

By Lemma 2, $d[\Pi^{*}] \leq d[\Pi ]$. If $d[\Pi^{*}] < d[\Pi ]$, then we take $\Pi' = \Pi^{*}$. If $d[\Pi^{*}] = d[\Pi ]$, then at least one of the derivations of the minor premisses has degree = $d[\Pi ]$. For the sake of the argument, let's assume that both have degree = $d[\Pi ]$. By the induction hypothesis, the derivation
\small{

\begin{prooftree}
\AxiomC{$[\cdot ;A]$}
\AxiomC{$\Gamma_{2}$}
\noLine
\BinaryInfC{$\Pi_{2}$}
\noLine
\UnaryInfC{$\Delta_{2} ; C$}
\AxiomC{$\Sigma_{1}$}
\noLine
\UnaryInfC{$\Theta_{1};\Lambda_{1}$}
\AxiomC{$\ldots$}
\AxiomC{$\Sigma_{m}$}
\noLine
\UnaryInfC{$\Theta_{m};\Lambda_{m}$}
\QuaternaryInfC{$\Delta_{2} , \Theta_{1} ,\ldots, \Theta_{m} ; \Lambda$}
\end{prooftree}
}

reduces to a derivation $\Pi_{2}'$ of $\Delta_{2} , \Theta_{1} , \ldots, \Theta_{m} ; \Lambda$ such that $d[\Pi_{2}']$ $<$ $d[\Pi_{2}]$, and the derivation

\begin{prooftree}
\AxiomC{$[\cdot ;B]$}
\AxiomC{$\Gamma_{3}$}
\noLine
\BinaryInfC{$\Pi_{3}$}
\noLine
\UnaryInfC{$\Delta_{3} ; C$}
\AxiomC{$\Sigma_{1}$}
\noLine
\UnaryInfC{$\Theta_{1};\Lambda_{1}$}
\AxiomC{$\ldots$}
\AxiomC{$\Sigma_{m}$}
\noLine
\UnaryInfC{$\Theta_{m};\Lambda_{m}$}
\QuaternaryInfC{$\Delta_{3} , \Theta_{1} ,\ldots, \Theta_{m} ; \Lambda$}
\end{prooftree}

Reduces to a derivation $\Pi_{3}'$ of $\Delta_{3} , \Theta_{1} ,\ldots, \Theta_{m} ; \Lambda$ such that $d[\Pi_{3}']$ $<$ $d[\Pi_{3}]$. \\

Let $\Pi'$ be:
\begin{prooftree}
\AxiomC{$\Gamma_{1}$}
\noLine
\UnaryInfC{$\Pi_{1}$}
\noLine
\UnaryInfC{$\Delta_{1} ; (A \vee_{i} B)$}
\AxiomC{$\Pi_{2}'$}
\noLine
\UnaryInfC{$\Delta_{2} , \Theta_{1} ,\ldots, \Theta_{m} ; \Lambda$}
\AxiomC{$\Pi_{3}'$}
\noLine
\UnaryInfC{$\Delta_{3} , \Theta_{1} ,\ldots, \Theta_{m} ; \Lambda$}
\TrinaryInfC{$\Delta_{1} , \Delta_{2} , \Delta_{3} , \Theta_{1} ,\ldots, \Theta_{m} ; \Lambda$}
\end{prooftree}
We can easily see that $\Pi$ reduces to $\Pi'$ and that $d[\Pi' ]$ $<$ $d[\Pi ]$. 

\item $\Pi$ is
{\small
\begin{prooftree}
\AxiomC{$\Pi_{1}$}
\noLine
\UnaryInfC{$\Delta_{1}, A, A:C$}
\RightLabel{$C_{c}$}
\UnaryInfC{$\Delta_{1}, A:C$}
\AxiomC{$\Sigma_{1}$}
\noLine
\UnaryInfC{$\Theta_{1};\Lambda_{1}$}
\AxiomC{$\ldots$}
\AxiomC{$\Sigma_{m}$}
\noLine
\UnaryInfC{$\Theta_{m};\Lambda_{m}$}
\QuaternaryInfC{$\Delta_{1} , \Delta_{2} , \Delta_{3} , \Theta_{1} ,\ldots, \Theta_{m} ; \Lambda$}
\end{prooftree}
}
By a permutative reduction, $\Pi$ reduces to the following derivation $\Pi^{*}$:\\
{\small
\begin{prooftree}
\AxiomC{$\Pi_{1}$}
\noLine
\UnaryInfC{$\Delta, A, A:C$}
\AxiomC{$\Sigma_{1}$}
\noLine
\UnaryInfC{$\Theta_{1};\Lambda_{1}$}
\AxiomC{$\ldots$}
\AxiomC{$\Sigma_{m}$}
\noLine
\UnaryInfC{$\Theta_{m};\Lambda_{m}$}
\QuaternaryInfC{$\Delta, A, A, \Theta_{1} ,\ldots, \Theta_{m}; \Lambda$}
\RightLabel{$C_{c}$}
\UnaryInfC{$\Delta , A,  \Theta_{1} ,\ldots, \Theta_{m}, A; \Lambda$}
\end{prooftree}
}
As in the previous case, by Lemma 2, $d[\Pi^{*}] \leq d[\Pi ]$. If $d[\Pi^{*}] < d[\Pi ]$, then we take $\Pi' = \Pi^{*}$. If $d[\Pi^{*}] = d[\Pi ]$, then by the induction hypothesis, the derivation
{\small
\begin{prooftree}
\AxiomC{$\Pi_{1}$}
\noLine
\UnaryInfC{$\Delta, A, A:C$}
\AxiomC{$\Sigma_{1}$}
\noLine
\UnaryInfC{$\Theta_{1};\Lambda_{1}$}
\AxiomC{$\ldots$}
\AxiomC{$\Sigma_{m}$}
\noLine
\UnaryInfC{$\Theta_{m};\Lambda_{m}$}
\QuaternaryInfC{$\Delta, A, A, \Theta_{1} ,\ldots, \Theta_{m}; \Lambda$}
\end{prooftree}
}
reduces to a derivation $\Pi^{**}$ of $\Gamma^{*} \subseteq \Gamma \vL \Delta, A, A, \Theta_{1} ,\ldots, \Theta_{m}; \Lambda$ such that $d[\Pi^{**}] < d[\Pi ]$. We can then take the desired derivation $\Pi'$ as
\begin{prooftree}
\AxiomC{$\Pi^{**}$}
\noLine
\UnaryInfC{$\Delta, A, A, \Theta_{1} ,\ldots, \Theta_{m}; \Lambda$}
\UnaryInfC{$\Delta, A, \Theta_{1} ,\ldots, \Theta_{m}; \Lambda$}
\end{prooftree}
\end{enumerate}
\end{itemize}
\end{proof}

\begin{lemma}\label{lemma:main}
Let $\Pi$ be a derivation of $\Gamma\vL \Delta ; \Sigma$ with $d[\Pi]>0$.
Then, $\Pi$ reduces to a derivation $\Pi'$ of $\Gamma'\vL \Delta ; \Sigma$ with $\Gamma ' \subseteq \Gamma$, such that $d[\Pi' ]$ $<$ $d[\Pi ]$. 
\end{lemma}
\begin{proof}
By induction on the length of $\Pi$.
\begin{itemize}
\item Case 1: $\Pi$ ends with an application of an introduction rule. \\
This case follows directly from the induction hypothesis.
\item Case 2: $\Pi$ ends with an application of an elimination rule. The general form of $\Pi$ is:
\begin{prooftree}
\AxiomC{$\Pi_{1}$}
\AxiomC{$\Pi_{2}$}
\AxiomC{$\Pi_{3}$}
\TrinaryInfC{$\Delta ; \Lambda$}
\end{prooftree}
By the induction hypothesis, $\Pi_{k}$ reduces to a derivation $\Pi_{k}'$ such that $d[\Pi_{k}' ]$ $<$ $d[\Pi_{k} ]$ ($1 \leq k \leq 3$).\\
Let $\Pi^{*}$ be:
\begin{prooftree}
\AxiomC{$\Pi_{1}'$}
\AxiomC{$\Pi_{2}'$}
\AxiomC{$\Pi_{3}'$}
\TrinaryInfC{$\Delta ; \Lambda$}
\end{prooftree}
By Lemma~\ref{lemma:reduction}, $d[\Pi^{*}]$ $\leq$ $d[\Pi ]$. If $d[\Pi^{*}]$ $<$ $d[\Pi ]$, we take $\Pi^{*} = \Pi'$. If $d[\Pi^{*}] = d[\Pi ]$, the $\Pi^{*}$ is a critical derivation and the result follows from Lemma~\ref{lemma:critical}.

\end{itemize}
\end{proof}

\begin{theorem} (\textbf{Normalization Theorem} for $\mathcal{LE}_{p}$)
Let $\Pi$ be a derivation of $\Gamma\vL \Delta ; \Sigma$ in $\mathcal{LE}_p$. Then, $\Pi$ reduces to a normal derivation $\Pi'$ of $\Gamma'\vL \Delta ; \Sigma$ where $\Gamma ' \subseteq \Gamma$.
\end{theorem}
\begin{proof}
 Directly from Lemma~\ref{lemma:main} by induction on $d[\Pi ]$.
\end{proof}

\section{Concluding remarks and future work}
There are lots of things to be done in the domain of ecumenical systems and more specifically in connection with pure ecumenical systems. We conclude this paper by mentioning  two possible lines of work we are pursuing. 

\subsection{Pure First-order ecumenical systems}
We can easily show that, in Prawitz' ecumenical system, if the main operator of a formula $A$ is classical, if $\Gamma , \neg A \vdash \bot$, then $\Gamma \vdash A$. For example, assume a derivation $\Pi$ as follows:

\begin{prooftree}
\AxiomC{$\Gamma$}
\AxiomC{$\neg (B \vee_{c} C)$}
\noLine
\BinaryInfC{$\Pi$}
\noLine
\UnaryInfC{$\bot$}
\end{prooftree}

Then, we can construct the following derivation of $\Gamma \vdash (B \vee_{c} C)$:

\begin{prooftree}
\AxiomC{$\Gamma$}
\AxiomC{$[(B \vee_{c} C)]^{1}$}
\AxiomC{$[\neg B]^{2}$}
\AxiomC{$[\neg C]^{3}$}
\TrinaryInfC{$\bot$}
\RightLabel{1}
\UnaryInfC{$[\neg (B \vee_{c} C)]$}
\noLine
\BinaryInfC{$\Pi$}
\noLine
\UnaryInfC{$\bot$}
\RightLabel{$2,3$}
\UnaryInfC{$(B \vee_{c} C)$}
\end{prooftree}

The same result holds for the classical existential quantifier. From the derivation

\begin{prooftree}
\AxiomC{$\Gamma$}
\AxiomC{$\neg  \exists_{c} xB(x)$}
\noLine
\BinaryInfC{$\Pi$}
\noLine
\UnaryInfC{$\bot$}
\end{prooftree}

We can easily obtain the derivation:

\begin{prooftree}
\AxiomC{$\Gamma$}
\AxiomC{$[\exists_{c} x(Bx)]^{1}$}
\AxiomC{$[\forall x\neg B(x)]^{2}$}
\BinaryInfC{$\bot$}
\LeftLabel{1}
\UnaryInfC{$[\neg \exists_{c} c(B(x)]$}
\noLine
\BinaryInfC{$\Pi$}
\noLine
\UnaryInfC{$\bot$}
\LeftLabel{2}
\UnaryInfC{$ \exists_{c} xB(x)$}
\end{prooftree}

\noindent From these results it follows that the classical implication $\to_{c}$ satisfies \textit{modus ponens} for classical succedents:

\begin{theorem} 
Let the main operator of the formula $B$ be classical. Then we can prove in Prawitz' system that $\{ (A \to_{c} B), A \} \vdash B$.
\end{theorem}

\noindent In the propositional system with the \textit{stoup}, we can  obtain the same results for the classical propositional operators. For example, given a derivation of

\begin{prooftree}
\AxiomC{$\cdot ; \neg (B \vee_{c} C)$}
\AxiomC{$\Gamma$}
\noLine
\BinaryInfC{$\Pi$}
\noLine
\UnaryInfC{$\Delta ; \cdot$}
\end{prooftree}

We can obtain a derivation of 

\begin{prooftree}
\AxiomC{$\Gamma$}
\noLine
\UnaryInfC{$\Pi$}
\noLine
\UnaryInfC{$\Delta ; (B \vee_{c} C)$}
\end{prooftree}

\noindent  But in the first-order case, a \textit{pure rule} for the classical existential quantifier in the {\em stoup} format requires an extra attention. Consider the following pure rules with {\em stoup} for the classical existential quantifier:

\begin{prooftree}
\AxiomC{$\Delta_{1} , A(t); \cdot$}
\UnaryInfC{$\Delta ; \exists_{c} xA(x)$}
\DisplayProof
\qquad \qquad
\AxiomC{$\Delta_{1} ; \exists_{c} xA(x)$}
\AxiomC{$\cdot ;A(a)$}
\noLine
\UnaryInfC{$\Pi$}
\noLine
\UnaryInfC{$\Delta_{2} ;\cdot$}
\BinaryInfC{$\Delta_{1} ; \Delta_{2} ;\cdot$}
\end{prooftree}

\noindent It is easy to show that the first order system with {\em stoup} obtained by means of the addition of the intuitionistic rules for $\exists_{i}$, $\forall$ and these rules for $\exists_{c}$ is not complete with respect to Prawitz' first-order ecumenical natural deduction. The important relation $\neg  \forall x\neg A(x) \vdash \exists_{c} xA(x)$ between $\forall$ and $\exists_{c}$ is not derivable in this first-order system with {\em stoup}. As a future work, we propose to investigate the first-order system obtained by the addition of the rules mentioned above plus a new structural rule, the \textit{store} rule:

\begin{center}
\begin{prooftree}
\AxiomC{$\Delta , A;\cdot$}
\RightLabel{{\em store}}
\UnaryInfC{$\Delta ;A$}
\end{prooftree}
\end{center}
with the side condition that the main operator of $A$ is classical.

\subsection{A different approach to \textit{purity}}
A different and interesting approach to pure systems worth exploring is based on some ideas proposed by Julien Murzi in~\cite{Murzi2018}. Murzi proposes a \textit{pure} single-conclusion Natural Deduction system that satisfies the basic inferentialist requirements of harmony and separability. Murzi's proposal combines (in a very interesting way!) Peter Schroeder-Heister's  idea of \textit{higher-level rules} with Neil Tennant's idea that the sign $\bot$ for the \textit{absurd} should be conceived as a \textit{punctuation mark}. Using Murzi's idea we can formulate a new \textit{pure} ecumenical natural deduction system for classical and intuitionistic logic.\\

(1) The \textit{impure} rule for $\vee_{c}$-Int becomes

\begin{prooftree}
\AxiomC{$[\frac{A}{\bot}]^{j}$}
\AxiomC{$[\frac{B}{\bot}]^{k}$}
\noLine
\BinaryInfC{$\underbrace{}$}
\noLine
\UnaryInfC{$\vdots$}
\noLine
\UnaryInfC{$\bot$}
\LeftLabel{$j,k$}
\UnaryInfC{$(A \vee_{c} B)$}
\end{prooftree}

(2) $\vee_{c}$-Elim.

\begin{prooftree}
\AxiomC{$(A \vee_{c} B)$}
\AxiomC{$[A]^{j}$}
\noLine
\UnaryInfC{$\vdots$}
\noLine
\UnaryInfC{$\bot$}
\AxiomC{$[B]^{k}$}
\noLine
\UnaryInfC{$\vdots$}
\noLine
\UnaryInfC{$\bot$}
\LeftLabel{$j,k$}
\TrinaryInfC{$\bot$}
\end{prooftree}

(3) The \textit{impure} rule for $\to_{c}$-Int becomes

\begin{prooftree}
\AxiomC{$[A]^{j}$}
\AxiomC{$[\frac{B}{\bot}]^{k}$}
\noLine
\BinaryInfC{$\underbrace{}$}
\noLine
\UnaryInfC{$\vdots$}
\noLine
\UnaryInfC{$\bot$}
\LeftLabel{$j,k$}
\UnaryInfC{$(A \to_{c} B)$}
\end{prooftree}

(4) $\to_{c}$-Elim.

\begin{prooftree}
\AxiomC{$(A \to_{c} B)$}
\AxiomC{$A$}
\AxiomC{$[B]^{k}$}
\noLine
\UnaryInfC{$\vdots$}
\noLine
\UnaryInfC{$\bot$}
\LeftLabel{$k$}
\TrinaryInfC{$\bot$}
\end{prooftree}

It is easy to show that the impure rules can be obtained from the new pure rules. In the case of $\vee_{c}$, for example, given a derivation $\Pi$ of $\bot$ from $\neg A$ and $\neg B$, we can construct the following derivation:
\begin{prooftree}
\AxiomC{$[\frac{A}{\bot}]^{j}$}
\UnaryInfC{$[\neg A]$}
\AxiomC{$[\frac{B}{\bot}]^{k}$}
\UnaryInfC{$[\neg B]$}
\noLine
\BinaryInfC{$\Pi$}
\noLine
\UnaryInfC{$\bot$}
\LeftLabel{$j,k$}
\UnaryInfC{$(A \vee_{c} B)$}
\end{prooftree}

In order to prove the other direction, it is convenient to add a new general rule that allows us to conclude rules. In the formulation of the rule we will use (as Murzi does) the expression $\Delta/A$ as an alternative to the rule $\frac{\Delta }{A}$.

\begin{prooftree}
\AxiomC{$[A]^{j}$}
\noLine
\UnaryInfC{$\Pi$}
\noLine
\UnaryInfC{$B$}
\LeftLabel{$j$}
\UnaryInfC{$(A /B)$}
\end{prooftree}

Suppose now that we have a derivation 

\begin{prooftree}
\AxiomC{$\frac{A}{\bot}$}
\AxiomC{$\frac{B}{\bot}$}
\noLine
\BinaryInfC{$\Pi$}
\noLine
\UnaryInfC{$\bot$}
\end{prooftree}

We can then construct the following derivation:
\begin{prooftree}
\AxiomC{$[A]^{1}$}
\AxiomC{$[\neg A]^{3}$}
\BinaryInfC{$\bot$}
\LeftLabel{1}
\UnaryInfC{$[(A /\bot )]$}
\AxiomC{$[B]^{2}$}
\AxiomC{$[\neg B]^{4}$}
\BinaryInfC{$\bot$}
\LeftLabel{2}
\UnaryInfC{$[(B /\bot )]$}
\noLine
\BinaryInfC{$\Pi$}
\noLine
\UnaryInfC{$\bot$}
\LeftLabel{$3,4$}
\UnaryInfC{$(A \vee_{c} B)$}
\end{prooftree}

An example: \textit{Peirce's law}:

\begin{prooftree}
\AxiomC{$((A \to_{c} B) \to_{c} A)^{2}$}
\AxiomC{$[A]^{1}$}
\AxiomC{$[\frac{A}{\bot}]^{3}$}
\BinaryInfC{$\bot$}
\LeftLabel{1}
\UnaryInfC{$(A \to_{c} B)$}
\AxiomC{$[\frac{A}{\bot}]^{3}$}
\TrinaryInfC{$\bot$}
\LeftLabel{$2,3$}
\UnaryInfC{$(((A \to_{c} B) \to_{c} A) \to_{c} A)$}
\end{prooftree}

\end{document}